\documentclass[runningheads]{llncs}
\usepackage[pdftex]{graphicx}
\usepackage{amsmath}
\usepackage{tikz}
\usepackage{algorithm,algpseudocode}
\newtheorem{Lemma}{Lemma}

\newtheorem{Theorem}[Lemma]{Theorem}

\begin{document}
\title{Sorting permutations with a transposition tree}
%
%
\author{ Bhadrachalam Chitturi \inst{1,3} \and
Indulekha T S \inst{2}}
\authorrunning{B. Chitturi et al.}
%
%
%
\institute{Department of Computer Science and Engineering,\\ Amrita Vishwa Vidyapeetham, Amritapuri Campus, India
 \and
 Department of Computer Science and Applications,\\ Amrita Vishwa Vidyapeetham, Amritapuri Campus, India
\and
Dept. of C S, University of Texas at Dallas, Richardson, TX
\email{bhadrachalam@am.amrita.edu, indulekhats@am.amrita.edu }
}
\maketitle              
\begin{abstract}
The set of all permutations with $n$ symbols is a symmetric group denoted by $S_n$.
A transposition tree, $T$, is a spanning tree over its $n$ vertices $V_T=${$1, 2, 3, \ldots n$} where the vertices are the positions of a permutation $\pi$ and $\pi$ is in $S_n$. 
$T$ is the operation and the edge set $E_T$ denotes the corresponding generator set. 
The goal is to sort a given permutation $\pi$ with $T$.
The number of generators of $E_T$ that suffices to sort any $\pi \in S_n$ constitutes an upper bound. 
It is an upper bound, on the diameter of the corresponding Cayley graph $\Gamma$ i.e. $diam(\Gamma)$. A precise upper bound equals $diam(\Gamma)$. Such bounds are known only for a few tress. Jerrum showed that computing $diam(\Gamma)$ is intractable in general if the number of generators is two or more whereas $T$ has $n-1$ generators.
For several operations computing a tight upper bound is of theoretical interest. Such bounds have applications in evolutionary biology to compute the evolutionary relatedness of species and parallel/distributed computing for latency estimation. 
The earliest algorithm computed an upper bound $f(\Gamma)$ in a $\Omega(n!)$ time by examining all $\pi$ in $S_n$. Subsequently, polynomial time algorithms were designed to compute upper bounds or their estimates. We design an upper bound $\delta^*$ whose cumulative value for all trees of a given size $n$ is shown to be the tightest  for $n \leq 15$. We show that $\delta^*$ is tightest known upper bound for full binary trees. \footnote{LNCS style}
\end{abstract}
\section{Introduction}

A \emph{transposition tree} $T=(V_T,E_T)$ is a spanning tree over $V_T$ where the cardinality of $V_T$ is $n$ and
$V_T=(1,2,\ldots n)$ \cite{Hey:2001}. 
The set of permutations with $n$ symbols forms a symmetric group denoted by $S_n$. 
Let $\pi \in S_n$ be the permutation that needs to be sorted. Let $E_T$ be the  set of edges of $T$. 
Let $\pi(i)$ and $V^i_T$ (or simply $i$) denote $i^{th}$ symbol of $\pi$ and the  $i^{th}$ vertex respectively. We employ the notation from \cite{Chitturi:2013,Ganesan:IAENG,A:K:1989,Hey:2001}.
The symbol $\pi(i)$ resides at vertex $i$ and is called as a \emph{marker}. If $\pi(i)=j$ then position $j$ is \emph{home} for the marker $\pi(i)$. Note that $j$ is the rank of $\pi(i)$ when $\pi$ is sorted ascending. 
An edge $(i,j) \in E_T$ signifies that $\pi(i)$ and $\pi(j)$ can be interchanged, i.e. swapped. A \emph{move} refers to one such interchange. 
$d_T(\pi)$ denotes the minimum number of moves to transform $\pi$ into the identity permutation $I_n=(1, 2, 3, 4, \ldots n)$ employing the generator set $E_T$. 
Since there is a path between every pair of vertices in a spanning tree, any two symbols of $T$ can be swapped by a sequence of moves. Thus, $T$ generates the symmetric group \cite{Smith:1999}.

Let $S$ be a group and $G$ be the associated set of generators. The Cayley graph $\Gamma$ of $S$ and $G$ is a graph having one vertex for each member of $S$ and an edge $(a,b) \rightarrow \exists g \in G$ such that $ag=b$. Given a transposition tree $T$, $E_T$ denotes the set of generators. A specific edge $e=(i,j)\in E_T$ is a specific generator. Application of some $e$ is a \emph{move}. 
 Let $\Gamma$ be the Cayley graph of $T$. An upper bound is a value $k$ such that any $\pi$ in $S_n$ can be sorted in at most $k$ moves. 
 Thus, any two vertices in  $\Gamma$ are at most $k$ edges apart.
The \emph{distance}  between $V^i_T$  and $V^j_T$  in $T$ is represented by $d_T(i,j)$. 
$diam(\Gamma)$, the diameter equals the maximum distance between any pair of vertices. 
An exact upper bound to sort any $ \pi \in S_n $ equals $diam(\Gamma)$. In this article we seek to compute an upper bound for $diam(\Gamma)$.

Certain Cayley graphs, were shown \cite{A:K:1989} to have a diameter that is sub-logarithmic 
in  the number of vertices, $n!$. For the prefix reversal operation $diam(\Gamma)$ is known to be linear; the best known upper bound is $18n/11$ \cite{Chitturi:etal:2009}. For prefix transposition the best known upper bound is $n - \log_{7/2}n$ \cite{Chitturi:2015}.
Thus, Cayley networks replaced hypercubes whose $diam(\Gamma)$ equals $\log (|V_T|)$) as the choice for interconnection networks and they have additional properties like vertex symmetry \cite{A:K:1989,Hey:1997,LV:etal:1993}. 
 Jerrum showed that the problem of identifying a minimum length sequence of generators when the number of generators is $\geq 2$ is intractable \cite{Jerrum:1985}. So, in general, given some $T$ on $n$ vertices, the computation of $diam(\Gamma)$ is NP-hard and efficient computation of tight upper bound for $diam(\Gamma)$ is sought. The research in the area of Cayley graphs has been active \cite{cdas2018,Kim201689,Konstantinova2016,chitturi2016adjacencies}. Recently, cube-connected circulants topology is shown to be better than some well-known network topologies \cite{mokhtar2016few}. In addition to permutations, such distance measures and their upper bounds have been extensively studied on strings \cite{Fertin:2009,Chitturi:2007,Chitturi:2008,Chitturi:2011}.

\section{Background}
Given a transposition tree  $T_n=(V_T,E_T)$, $\forall_i e_i \in E_T$ is  a generator, and the application of any one of the $n-1$ generators is a \emph {move}.
For a given operation say flip (prefix reversal)  on $ \pi \in S_n$, a prefix of length $k$ where $2 \leq k \leq n$ is reversed corresponding to $n-1$ generators $k=2\ldots n$.
 
Determining good upper bounds for sorting permutations under various operations is of interest. The computation of exact upper bound for sorting permutations with many operations is either intractable or its complexity is unknown \cite{Fertin:2009}. 
We state some results in the general area of sorting permutations with various generator sets. In 2009 Chitturi \emph{et al.}  \cite{Chitturi:etal:2009} improved the upper bound given by Gates and Papadimitriou \cite{GP:1979} for sorting permutations with prefix reversals. 
The problem of computing the diameter of the Cayley graph generated by cyclic adjacent transpositions was introduced by Jerrum \cite{Jerrum:1985}; for which Feng et al. \cite{Chitturi:AEMB:2010} prove a lower bound of $n^2/4$. 
An $O(n^3)$ amortized time algorithm to compute the optimum number of moves to sort any permutation with transposition operation was designed \cite{cdas2018}.

 Akers and Krishnamurthy computed an upper bound  $f(\Gamma)$ for the $diam(\Gamma)$ for transposition trees in $\Omega(n! n^2)$ time \cite{A:K:1989}. Given a transposition tree $T$, Ganesan~\cite{Ganesan:IAENG} computes a non-deterministic measure $\beta$, an estimate of the exact upper bound.  $\beta_{max}$ is the maximum among all values of $\beta$. Only $\beta_{max}$ that requires exponential time to compute is an upper bound. It is shown that $\beta_{max}\leq f(\Gamma)$.
Kraft \cite{Kraft} proposes three algorithms to identify  upper bounds on $diam(\Gamma)$ generated using transposition trees: $\alpha$ an exponential time algorithm, $\eta$ and $\xi$ a randomized algorithm. The method $\alpha$ tries to improve the bound of  \cite{A:K:1989} by identifying the minimal value at each step.

The terminology used in  \cite{Chitturi:2013,A:K:1989,Ganesan:IAENG,Hey:2001} is adopted here. 
Let $T$ be transposition tree on a vertex set $\{1,2,\ldots,n\}$. Let  $\pi$ be the permutation in $S_n$ that is to be sorted.
If  $\pi =(7,6,5,4,3,2,1)$ then $\pi \in S_7$ and in cycle representation $\pi =(1,7) (2,6)(3,5)(4)$ that is, $\pi $ has four cycles i.e. $\eta(\pi)=4$ \cite{A:K:1989}. Likewise, if $\pi =(6,5,4,3,2,1) = (1,6) (2,5) (3,4)$ then  $\eta(\pi)=3$.
A \emph {marker} $\pi(i)$ resides at position $i$. 
The move which swaps elements at positions $i$ and $j$ is denoted as $(i~j)$. 
If $(i~j)$ is executed on $\pi$ then we obtain $\pi (i~j)$, the result of application of $(i~j)$ to $\pi$.
Application of a sequence of generators so that $\pi(j)$ reaches its home is called \emph{homing} $\pi(j)$.

The  diameter of the Cayley graph is identified only for some transposition trees. 
If $T$ is  $K_{1,n-1}$, a star graph, the diameter is  $\lfloor 3(n-1)/2 \rfloor$ \cite{A:K:1989} for  bubble sort graph, it equals ${^n C_2}$. Recently, \cite{Uthan:2017} identified diameter for two  novel classes of trees $S_{m,k}$  and $M_k$.
 $S_{m,k}$ has one central vertex \emph{center} and $m$ spokes, each spoke is a path of length $k+1$ where all the paths share  \emph{center}. Thus, the regular star tree with $k-1$ leaves i.e. $S_k$ or $S_{1,k-1}$ is same as $S_{k-1, 1}$. Let $f(S_{m,k})$ be the diameter of the Cayley graph generated by $S_{m,k}$ then $f(S_{m,k}) = mk(2k+1)/2$ \cite{Uthan:2017}. Matchstick tree,$M_k$ was defined in \cite{Chitturi:2010}. It is an extension of the path graph with $k$ vertices where each vertex in the path graph has a corresponding leaf attached to it. The diameter of a Cayley graph generated with $M_k$ is $k^2+k -1$ for $k>2$ \cite{Uthan:2017}.

One can adopt brute force search to compute the diameter corresponding to $T_n$ for $n<6$. In the current article $n \ge 6$. In Section 3 we show the computation of $\delta^*$ in $O(n^2)$ time. In Section 4 we derive expression for  $\delta^*$  and compare it to the best existing bounds. 

\section{Algorithm $\delta^*$}
$Eccentricity$ of a vertex $u$ in $V_T$, denoted by $ecc(u)$, equals the maximum value of $d_T(u,i)$ for all $i$. The center of a tree is either a vertex (centered tree) or an edge (bicentred tree). A unique vertex with minimum eccentricity forms the center of a tree. An adjacent pair of vertices (corr. to an edge)  with minimum eccentricity form the center of a bicentered tree.
%
Let $\mu$ be the upper bound computed by an algorithm $A$ on a given tree $T$. $Cum_{\mu}(n)$ is the cumulative value of  upper bounds of all trees with $n$ vertices computed by $A$. Likewise $Cum_{\mu}(X)$ is the cumulative value of upper bounds computed by $A$ for all trees belonging to class $X$.

Chitturi designed an algorithm \emph{Algorithm S}  that identifies the set of all vertices that have maximum eccentricity i.e. $S$ in $T$ in linear time \cite{Chitturi:2013}. The general idea of the algorithms in \cite{Chitturi:2013} is to delete a set of leaves say $X$ and obtain an upper bound on the number of moves that suffice to home markers to all vertices in $X$. 
A $cluster$ $C$ in $S$ is the maximal subset of $S$ such that any $u, v$ in $C$ are less than $diam(T)$ apart \cite{Chitturi:2013}. Note that if all but one of the clusters is deleted then the diameter of the resultant tree decreases. Based on this idea, Algorithm $D'$ computes an upper bound $\delta'$ that works by deleting all clusters in $S$ except the largest cluster $C^*$ \cite{Chitturi:2013}. It employs a linear time algorithm $Fast_{NC}$ to identify all clusters \cite{Chitturi:2013}. 
Algorithm $D'$ has two versions, $D'_{v_1}$ and $D'_{v_2}$ which calculates $\delta'_{v_1}$ and $\delta'_{v_2}$ respectively. $D'_{v1}$ removes the entire $S$ if $|S-C^*| > (2/3)|S|$ while $D'_{v_2}$ removes the entire S if $|S-C^*| >= (2/3)|S|$. 
It was shown that $D'_{v_1}$ algorithm has the best value for $Cum_{\delta'}(n)$ and $Cum_{\delta'}(B)$ where $B$ denotes full binary trees. Here we improve upon Algorithm $D'$ and design a new algorithm Algorithm $\delta^*$ that computes the new upper bound $\delta^*$. Results show that $Cum_{\delta'}(n) > Cum_{\delta^*}(n)$ for $n \in 6 \ldots 15$. The following observations form the basis of Algorithm $\delta^*$.\\\\
\emph{Observation 1}: If $|S \setminus C^*| > |C^*|$ then at most $|C^*|$ markers need $diam(T)$ moves each to be homed.\\
\emph{Proof}: The proof follows from pigeon hole principle. Let the set of vertices that are being deleted be $X$. 
If $|X|>|C^*|$ then at most $|C^*|$ from $C^*$ can be homed to the vertices of $X$ the rest must be from elsewhere including within $X$.\\


\begin {algorithm}[h]
\caption{\textbf{Algorithm  $\delta^*$}\label{algorithm:dp}}
\begin{algorithmic}[1]
\\ $T$ and $V_T$ are the current tree and the corr. set of vertices.
\\ \textbf{S1.}
	\State $\delta^* \leftarrow 0$
	\If {$T$ is  a star graph}
	 \State $\delta^* \leftarrow \delta^* + \lfloor 3/2(|V_T|-1)\rfloor $ and terminate.
	\EndIf 
\\ \textbf{S2.} 
	\State Identify $S$ by executing \emph{Algorithm S}.

\\ \textbf{S3.}  
	\State Compute clusters for $S$ with \emph{Algorithm NC}.

\\ \textbf{S4.} 
		\State Identify $C= S-C^*$ where $C^*$ is the largest cluster and $distSum(C^*)$ is the least  (break ties arbitrarily). $V_{C}$: vertices of $C$. $~diam(T \setminus V_{C})< diam(T) $. 
\State \textbf{Case 1}:  $|C^*| \geq |S|/2~$   	
        \State $\delta^* \leftarrow \delta^*+|C|*diam(T)$
        $T \leftarrow T \setminus V_C$ 
\State \textbf{End Case 1}
\State \textbf{Case 2}:  $|C^*| < |S|/2~~~~~~~~$   //$|C|\geq|C^*|$
		\State $\delta^* \leftarrow \delta^*+|C^*|*diam(T);~~T \leftarrow T \setminus V_C$ 
			\If {$|C|-|C*|$ is even}	\State $\delta^* \leftarrow \delta^*+ (|C|- |C^*|)*(diam(T)-1/2) $; 
			\Else \State $\delta^* \leftarrow \delta^*+ (|C|-|C^*|)*(diam(T)-1/2) -1/2 $;  
			\EndIf
\State \textbf{End Case 2}
\If {$T$ is not a star graph} go to step S2.
	\Else
	\State $\delta^* \leftarrow \delta^* + \lfloor 3/2(|V_T|-1)\rfloor $ and terminate.
\EndIf		
\end{algorithmic}
\end{algorithm}

%
%

\begin{Theorem} \label{thm:dprime:alg}
$\delta^*$ is an upper bound for sorting permutations.
\end{Theorem}
\begin{proof} 
The set of vertices that is to be deleted, i.e. $S \setminus C$ where $C$ is one of the largest clusters be called by $X$. $X$ can be union of several clusters and $X \subset S$ where $S$ is the set of vertices in $T$ with greatest eccentricity \cite{Chitturi:2013}. Recall that deletion implies that markers are being homed to the vertices in $X$. We try to obtain an upper bound on the cost to home all markers to $X$. The scenario yields Case 1) and Case 2). The other scenarios clearly yield a lower value as the markers are less that $diam(t)$ apart from their respective homes (destinations). 
Case 1): All markers homed to vertices in $X$ are from $X$. Case 2): Markers homed to vertices in $X$ are from any vertices of $S$.\\
Case 1):  
Lemma 5 of \cite{Chitturi:2013} gives an upper bound of $|X| (diam(T)-1/2)$ for Case 1). An additional half move is subtracted when $|X|$ is odd. \\
Case 2):  Let $|X|$ be $x$ and let $|C|$ be $c$. Let $i$ vertices from $C$ be homed to corresponding $i$ vertices in $X$. The associated cost is $i diam(T)$ moves. It follows that $x-i$ vertices from $X$ are homed among themselves where the upper bound for the associated cost is $(x-i) (diam(T)-1/2)$ due to Lemma 5 of \cite{Chitturi:2013}. When $x-i$ is odd then cost reduces by $1/2$.
Thus, the total cost is  $(x-i) (diam(T)-1/2) + i.diam(T)$ which maximizes when $i$ is maximized. That is, $i=c$ yielding $(x-c) (diam(T)-1/2) + c.diam(T)$.  When $x-c$ is odd then cost reduces by $1/2$. 
Further, consider the scenario where $u \in C$ is to be homed to $v \in X$ and $v \in X$ is to be homed to $w \in X$. Note that homing $u$ first moves $v$ one edge closer to its home. So, it requires at most $diam(T)-1$ moves. Thus, this scenario does not yield the worst case. Further, if there are dependencies such as the home of $a$ is $b$, the home of $b$ is $c$ etc. then the dependency that yields the worst case is shown to be mutual swap of pairs of markers; that is, the home of $a$ is $b$ and the home of $b$ is $a$ (Lemma 5 of \cite{Chitturi:2013}). So, given a subset $Y$ of $S$ where the vertices of $Y$ are to be homed within $Y$, the maximum number of pairs of $Y$ are swapped employing their corresponding sequence of moves. Thus, the theorem follows.
\end{proof}

\section{Results}
The cumulative sum of all upper bound values for all trees with up to 10 and 15 vertices was recorded for all the existing algorithms in \cite{Chitturi:2013} and \cite{Uthan:2017} respectively. $\delta'$ yielded the minimum value \cite{Chitturi:2013,Uthan:2017}. Our results show that $\delta^*$ yields a smaller value than $\delta'$ for the same. We obtained all non-isomorphic trees with a  given number of vertices from sagemath.org \cite{Sage}. 
The results are tabulated in Table 1.
We theoretically show that for a full binary tree $\delta^*$ indeed is tighter than $\delta'$. In \cite{Chitturi:2013} $\delta'$ is shown to be deterministic (the choices of deletion do not alter the value of the measure). We show that $\delta^*$ is also deterministic by proving that various sequences of deletions of vertices leads to isomorphic trees. So, the comparison is valid.
Table 2 shows the corresponding execution results.
\begin{table}[h!]
	\caption{Comparison of $\delta^*$ with other methods for all trees with $n$ vertices.}
\begin{center}
\begin{tabular}{|c|c|c|c|}
\hline
 No:of nodes & $Cum_{\delta'_{v1}}$ & $Cum_{\delta'_{v2}}$ & $Cum_{\delta^*}$\\ 
 \hline
 6 & 63 & 63 & 63 \\  
 \hline
 7 & 154 & 153 & 153 \\  
 \hline
 8 & 409 & 407 & 407 \\  
 \hline
 9 & 1032 & 1028 & 1027 \\  
 \hline
 10 & 2819 & 2809 & 2805 \\  
 \hline
 11 & 7401 & 7376 & 7361 \\  
 \hline
 12 & 20277 & 20222 & 20175 \\  
 \hline
 13 & 50032 & 49931 & 49820 \\  
 \hline
 14 & 152585 & 152285 & 151855 \\  
 \hline
 15 & 212841 & 212532 & 212217 \\  
 \hline
 \end{tabular}
 \end{center}
 \end{table}
 %
 %
 \begin{table}[h!]
	\caption{Comparison of $\delta^*$ with  $\delta'$ on a full binary tree with depth $d$.}
 \begin{center}
\begin{tabular}{ |c|c|c|c|c|c| } 
 \hline
 No:of nodes & No:of leaves & Depth & $\delta'v1$ & $\delta'v2$ & $\delta^*$ \\
 \hline
 3 & 2 & 1 & 3 & 3 & 3 \\ 
 \hline
 7 & 4 & 2 & 17 & 17 & 15\\ 
 \hline
 15 & 8 & 3 & 58 & 58 & 55\\ 
 \hline
 31 & 16 & 4 & 171 & 172 & 167 \\ 
 \hline
 63 & 32 & 5 & 460 & 461 &  453\\ 
 \hline
 127 & 64 & 6 & 1165 & 1168 & 1153 \\ 
 \hline
 255 & 128 & 7 & 2830 & 2833 & 2807 \\ 
 \hline
\end{tabular}
\end{center}
\end{table}
\newline Algorithm $\delta^*$  ensures that in every iteration, the maximum cost required to home a node is determined by the availability of room for homing that node. This is the qualitative improvement of this article.
Let $B(d)$ be a full binary tree of depth $d$ with levels $(1,2,3,\ldots d+1)$ where root is at level 1. Let $m$ be the number of leaf nodes in $B(d)$.
When the algorithm is run on $B(d)$ for $(2d-2)$ iterations then $B(d)$ will be transformed into a star graph(whose centre node has degree 3, i.e. $K_{1,3}$). 
In the first iteration, the set $S$ contains all leaf nodes and forms two clusters. One cluster contains all the leaf nodes of the left subtree and the other contains all leaf nodes in the right subtree. 
Since the distance sums of both clusters are same, any one of the two clusters can be considered as the largest cluster and the other can be deleted. 
Let the algorithm remove $m/2$ nodes of the right cluster each at a cost of $2d$ ($2d$ is the initial diameter). 
If the left cluster is chosen for deletion then the resultant tree is isomorphic to the resultant tree in the previous case and the cost is identical.
 In the second iteration, two clusters are formed. One cluster $C_1$ with $m/2$ leaf nodes of left subtree and another $C_2$  with $m/4$ leaf nodes of right subtree. The smaller of the clusters, i.e. $C_2$ is deleted. 
 So, in the second iteration $m/4$ nodes are removed each at a cost of $2d-1$(after the first iteration, the diameter is reduced by $1$). After two iterations, nodes in the levels $d+1$ and $d$ of the right subtree of $B(d)$ are removed. Let us call the new tree as $T(d)$. So, $\delta^*= A + B+ C$ where,\\
$A = $ Cost to convert $B(d)$ to $T(d)$ (first 2 iterations)\\
$B = $ Total cost for the next $2d-4$ iterations which results in $K_{1,3}$\\
$C = $ $\delta^*(K_{1,3})$\\
\textbf{A:} The total cost to convert $B(d)$ to $T(d)=(m/2)2d + (m/4)(2d-1) = m(6d-1)/4$.\\
\textbf{B:} It takes $(2d-4)$ iterations, i.e. (3  \ldots $(2d-2)$) to transform $T(d)$ to $K_{1,3}$. 
It consists of$(d-2)$ odd iterations and $(d-2)$ even iterations. 
In every odd iteration, the current tree $T$ with the current center $c$ can be descried as follows.  
There is an edge from $c$ to the center of the tree of the previous iteration that we call as \emph{upward edge}. If we imagine that the subtree connected through upward edge does not exist then we obtain a binary tree. Our terminology of left and right subtree is based on such presumption. Three clusters $C_1,C_2$ and $C_3$ are formed 
where $C_1$ and $C_2$ are the the left and right sutrees of the current tree and $C_3$ is the subtree that is connected to the center through upward edge.
In the beginning of an odd iteration, the centre is a single node (since the diameter is even). 
Both $C_1$ and $C_2$ contain same number of nodes and are of maximum size, any one of these clusters(along with the third cluster $C_3$) will be chosen for deletion. Let us assume that $C_2$ and $C_3$ are chosen for deletion. 
The nodes in $C_2$ will be deleted at a cost of $\textit {current diameter}$ per node and 
the nodes of $C_3$ will be deleted at a cost of $\textit{current diameter}-1/2$(because $|C_1|-|C_2|=0$). 
The trees that are obtained either by deleting $C_1$ and $C_3$ or by deleting $C_2$ and $C_3$ are isomorphic. 
It is shown in Figure 1.\\
\begin{figure}[h!]
	\centering
	\includegraphics[width=0.75\textwidth]{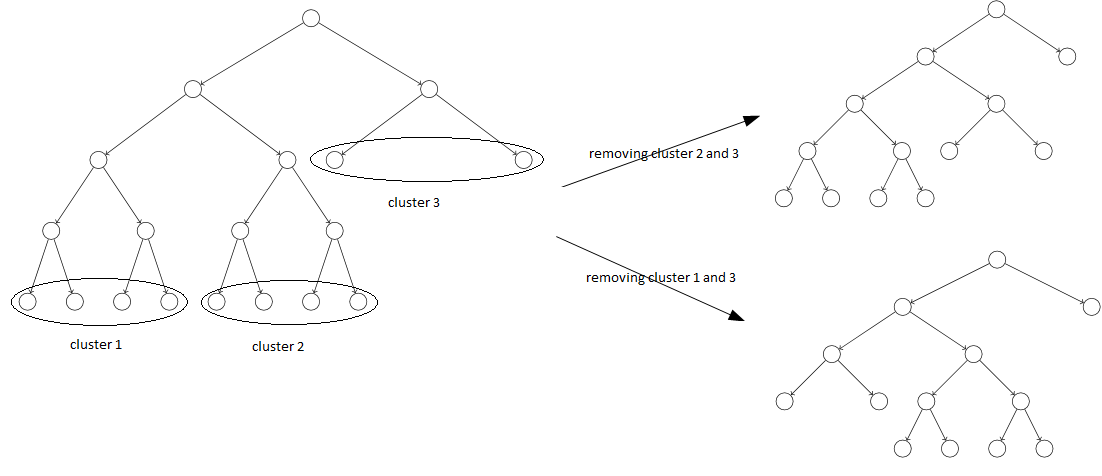}
	\caption{Removal of left or right cluster}
\end{figure}		
After the $d^{th}$ iteration, the diameter of the tree will be reduced to $d$ so that the entire right subtree is deleted(or the left, depends on the choice of the clusters, when there is a tie). In every odd iteration after this, all the three clusters will contain same number of nodes. Irrespective of the clusters chosen for deletion, the algorithm will generate the same(isomorphic) tree. Refer Figures 2 to 4. 
Every even iteration will generate exactly two clusters $C_1$ and $C_2$ with different sizes. Let us assume that $C_2$ is the smallest cluster and hence the one chosen for deletion.                                                                                               
\begin{figure}[h!]
	\centering
	\includegraphics[width=0.75\textwidth]{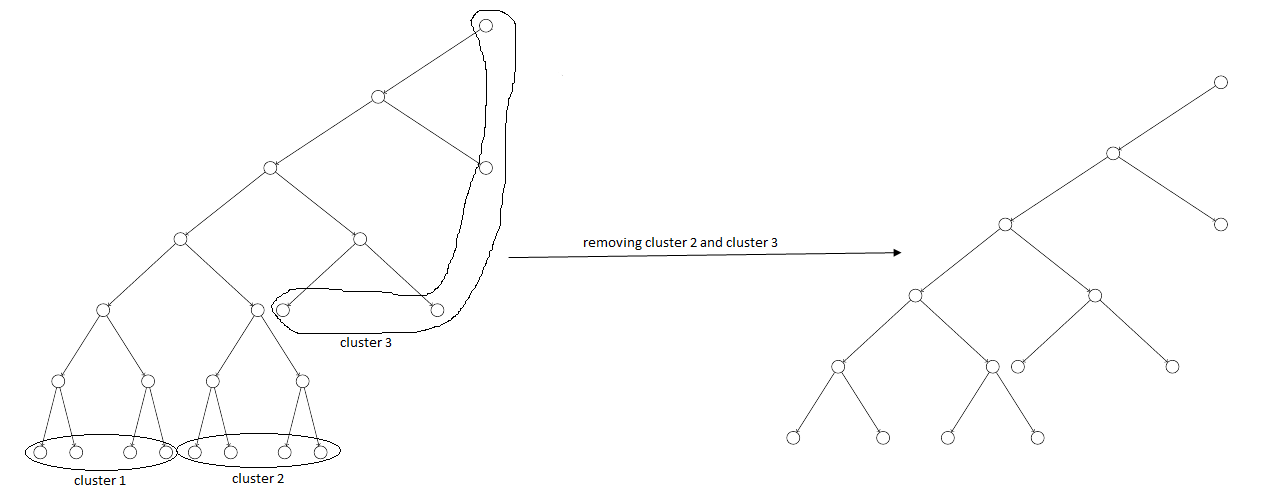}
	\caption{First possibility}
\end{figure}
\begin{figure}[h!]
	\centering
	\includegraphics[width=0.75\textwidth]{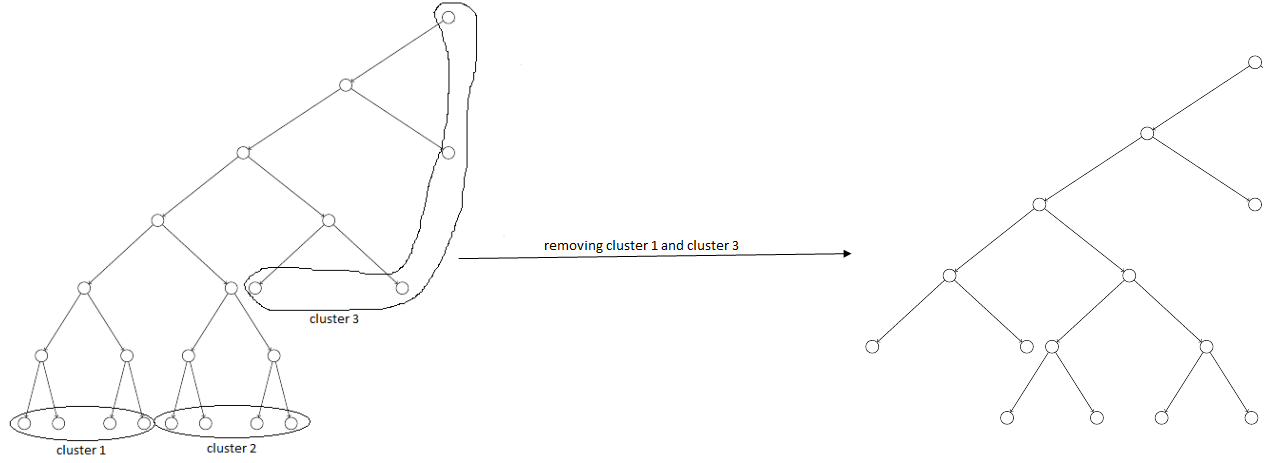}
	\caption{Second possibility}
\end{figure}
\begin{figure}[h!]
	\centering
	\includegraphics[width=0.75\textwidth]{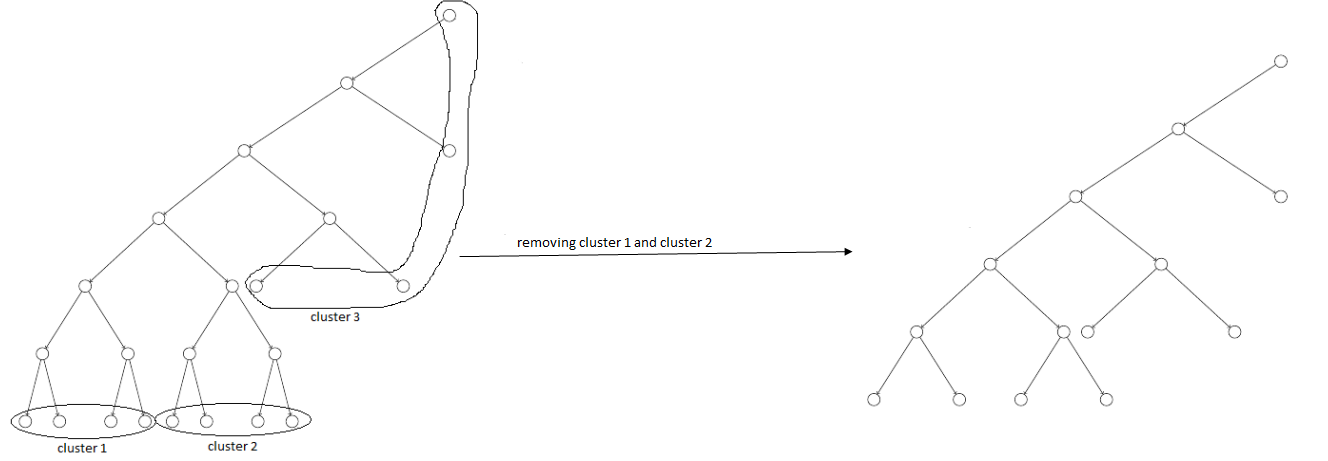}
	\caption{Third possibility}
\end{figure}
\\Table 3 and Table 4 show the number of nodes removed and the removal cost per node in various iterations of phase B. The removed nodes are categorized into two groups based on the removal cost per node.
\begin{table}[h!]
	\caption{Removal cost per node in the first $d$ iterations}
	\centering
	\begin{tabular}{c c c c c c}
		\hline
		Iteration & $|C_2|$ & Cost per node in $C_2$ & $|C_3|$  & Cost per node in $C_3$ \\
		\hline
		3 & $m/2^2$ & 2d-2  & $ m/2^3$ & $2d-2-1/2$ \\
		4 & $3m/2^4$ & 2d-3 & - & - \\
		5 & $m/2^3$ & 2d-4  & $ 3m/2^5$ & $2d-4-1/2$ \\
		6 & $7m/2^6$ & 2d-5 & - & - \\
		7 & $m/2^4$ & 2d-6  & $7m/2^7$ & $2d-6-1/2$\\
		8 & $15m/2^8$ & 2d-7& - & - \\
		... & ... & ... & ... & ...\\
		... & ... & ... & ... & ...\\                                                                      
		d &$ (2^{d/2}-1)m/2^d$ & $(d+1)$ & - & - &(for even $d$)\\
		d &$m/2^{(d+1)/2} $ & $(d+1)$ & $(2^{d/2}-1)m/2^d$ & $d-1/2$ &(for odd $d$)
	\end{tabular}
\end{table}\\
\begin{table}[h!]
	\caption{Removal cost per node in the next $d-2$ iterations}
	\centering
	\begin{tabular}{c c c c c c}
		\hline
		Iteration & $|C_2|$ & Cost per node in $C_2$ & $|C_3|$  & Cost per node in $C_3$ \\
		\hline
		... & ... & ... & ... & ...\\
		... & ... & ... & ... & ...\\
		$(2d-9)$  & 16 & 10 & 16 & 10-1/2\\
		$(2d-8)$  & 16 & 9 & - & -\\
		$(2d-7)$  & 8 & 8 &8 & 8-1/2\\
		$(2d-6)$  & 8 & 7 & - & -\\
		$(2d-5)$  & 4 & 6 & 4 & 6-1/2\\
		$(2d-4)$  & 4 & 5 & - & -\\
		$(2d-3)$  & 2 & 4 & 2 & 4-1/2\\
		$(2d-2)$  & 2 & 3 & - & -\\
	\end{tabular}
\end{table}\\
Thus, the total cost $B$ is the sum of three series $P$, $Q$, and $R$, where\\
$P = \dfrac{m(2d-2)}{2^2}+\dfrac{m(2d-4)}{2^3}+\dfrac{m(2d-6)}{2^4}+\ldots+\dfrac{m(2d-(2d-4)}{2^{d-1}}$\\
$Q = \dfrac{m(2d-2-1/2)}{2^3}+\dfrac{3m(2d-4-1/2)}{2^5}+\dfrac{7m(2d-6-1/2)}{2^7}+\ldots+\dfrac{(2^{d-2}-1)m(4-1/2)}{2^{2d-3}} $\\
$R = \dfrac{3m(2d-3)}{2^4}+\dfrac{7m(2d-5)}{2^6}+\dfrac{15m(2d-7)}{2^8   }+\ldots+\dfrac{2^{d-2}-1m(3)}{2^{2d-2}}$ \\
\textbf{C:} The cost required for the final star graph is $\left \lfloor \dfrac{3}{2}(4-1) \right \rfloor$ which is $4$.\\

\subsection{Comparing $\delta^*$ with $\delta'v_1$ and $\delta'v_2 $}
The improvement occurs when there are more number of nodes to be removed than that of the nodes in the largest cluster chosen. In every iteration, the largest cluster chosen by both $\delta'v_1$ and $\delta^*$ contain the same set of nodes.
Let $C$ is the largest cluster chosen. Since $|C| <=\dfrac{1}{2}|S| < \dfrac{2}{3}|S|$, both $\delta'v_1$ and $\delta^*$ will delete the set $S-C$ and generate the same trees(isomorphic) after every iteration. 
$\delta'v_1$ deletes each node at a cost of current diameter, while $\delta^*$ deletes some of them in lesser cost. In every odd iteration in phase B, there is an improvement of $0.5$ per node in the second cluster deleted.
The improvement of $\delta^*$ over $\delta'v_1$ is as follows:\\
$\delta^* = \delta'_{v_1} -\dfrac{(m+1)}{6}-1-\dfrac{1}{2} $ ~~~~//if 'd' is odd\\
$\delta^* = \delta'_{v_1} -\dfrac{(m+2)}{6}-1 $ ~~~~~~~~~//if 'd' is even\\

\begin{table}[h!]
	\caption{Improvement of $\delta^*$ over $\delta'v_1$ }
	\centering
	\begin{tabular}{c c c c c c}
		\hline 
		Iteration & Improvement \\
		\hline
		3 & $ m/2^4$ \\                                                      
		5 & $ 3m/2^6$ \\
		7 & $ 7m/2^8$ \\
		... & ... \\
		... & ... \\
		$(2d-9)$  & 8\\
		$(2d-7)$  & 4\\
		$(2d-5)$  & 2\\
		$(2d-3)$  & 1\\
	\end{tabular}
\end{table}
The behaviour of $\delta'v_2$ will be different from $\delta'v_1$, when $|C^*|\leq \dfrac{1}{2}|S - C^*|$.
In a full binary tree, from $(d+2)^{nd}$ or $(d+1)^{st}$ iteration onwards,(for odd and even depth trees respectively) the above condition holds. From that iteration onwards, the nodes removed by $i^{th}$ and $(i+1)^{th}$ iteration of $\delta'v_1$ will be removed by $\delta'v_2$ in a single iteration with lesser cost. This improvement is same as that of $\delta^*$ in those iterations.\\
Recall that either 3 or 2 clusters will be formed in every odd and even iteration respectively. Let the clusters in the $i^{th}$ odd iteration be $C_{1,i}$,$C_{2,i}$,and $C_{3,i}$, and that in the $(i+1)^{th}$ even iteration be $C_{1,i+1}$ and $C_{2,i+1}$. As discussed before, from this $i^{th}$ iteration onwards, all the three clusters formed are of same size.
In the $i^{th}$ iteration, $S = |C_{1,i}| + |C_{2,i}| + |C_{3,i}|$. Any of the clusters can be chosen as $C^*$. Then, $|C| = 2|S|/3$. So, $\delta'v_2$ will delete the entire $S$, with a cost of $(current diameter - 1/2)$ per node. Let $k$ denotes the current diameter. Then, the cost in $i^{th}$ iteration of $\delta'v_2$ is 
$cost_{i, \delta'v_1} = (k-1/2)(|C_{1,i}| + |C_{2,i}| + |C_{3,i}|) = 3|C_{1,i}|(k-1/2)$.\\
In the corresponding iteration, $\delta^*$ will remove only two clusters, say $C_{2,i}$,and $C_{3,i}$.So, the cost in this iteration is, 
$cost_{i, \delta*} = k(|C_{2,i}|) + (k-1/2)(|C_{3,i}|)$.
In the next iteration of $\delta^*$, it will form two clusters of equal size, say $C_{1,i+1}$,$C_{2,i+1}$. Since $C_{1,i}$ and $C_{1,i+1}$ contains the same set of nodes, removal of $C_{1,i+1}$ in this iteration will result the same tree as that in the $i^{th}$ iteration of $\delta'v_1$. The cost in this iteration is 
$cost_{i+1, \delta*} = (k-1)|C_i|$. 
So, the total cost of $\delta^*$ in $i^{th}$ and $(i+1)^{th}$ iteration is $3|C_{1,i}|(k-1/2)$, which is same as that of $\delta'v_2$ in $i^{th}$ iteration. \\

The final expressions obtained for each of these measures for $B(d)$: a full binary tree of depth $d$ are shown below.\\
$\delta^* = \delta'v_2 -\dfrac{(m+1)}{6}-\sqrt{\dfrac{m}{8}}-\dfrac{1}{2}$~~~~~~~~ //if d is odd\\\\
$\delta^* = \delta'v_2 -\dfrac{(m+2)}{6}-\sqrt{\dfrac{m}{4}} $ ~~~~~~~~~~~~~~//if d is even\\\\
$\delta^* = 4md - 6m + O(d^2) - \dfrac{(m+4)}{6} ;$~~~~//if $d$ is odd\\\\
$\delta^* = 4md - 6m + O(d^2) - \dfrac{(m-4)}{6} ;$~~~~//if $d$ is even\\\\
$\delta'_{v_1} = 4md - 6m + O(d^2)$;\\\\
$\delta'_{v_2} = 4md - 6m + O(d^2) + \sqrt{\dfrac{m}{8}} ;$~~~~~~~~//if $d$ is odd\\\\
$\delta'_{v_2} = 4md - 6m + O(d^2) + 1 + \sqrt{\dfrac{m}{4}} ;$~~~//if $d$ is even\\\\

The above expressions show that the gain for $\delta^*$ over the other two algorithms for a $B(d)$ is approximately $m/6$ where $m$ is the number of leaf nodes and the total number of nodes is $n$. However, $m=(n+1)/2$. Thus, gain is approximately $n/12$. The theorem stated below immediately follows. 
\begin{Theorem} \label{thm:d*better2}
$\delta'_{v_1}-\delta^* = \Omega(n)$ for a full binary tree.
\end{Theorem}

\section{Conclusion}
We design a new measure $\delta^*$. It is shown to yield the smallest cumulative sum for all trees of a given cardinality where $|V_T| \leq  15$. This is an improvement over the previously known best upper bound $\delta'$. For a full binary tree we show that $\delta^*$ is tighter than $\delta'$. In addition to full binary trees, were able to theoretically show that $\delta^*$ is tighter than $\delta'$ for several other tree classes.

\section{Acknowledgements}
Authors thank Amma for the direction.

\end{document}